\documentclass{article}
\usepackage[affil-it]{authblk}
\usepackage{fullpage, color-edits}
\usepackage{natbib}
\usepackage{amsthm, amsmath, amsfonts, amssymb}
\usepackage{hyperref}
\usepackage[linesnumbered, ruled,vlined]{algorithm2e}
\usepackage{color}
\usepackage{xspace}
\usepackage{tikz}
\usetikzlibrary{arrows,decorations,decorations.shapes,backgrounds,shapes}

\newcommand{\argmax}{\textrm{argmax}}
\renewcommand{\l}{\mathpzc{l}}

\newcommand{\lt}{\ensuremath{\mathpzc{l}_t}}
\newcommand{\prob}{\mathbb{P}}
\renewcommand{\exp}{\mathbb{E}}
\newcommand{\D}{\ensuremath{\mathcal{D}}\xspace}
\newcommand{\Di}{\ensuremath{\mathcal{D}_i}\xspace}

\newcommand{\Fei}{\ensuremath{\widehat{F}_i}\xspace}
\newcommand{\Fi}{\ensuremath{F_i}\xspace}

\renewcommand{\v}{\ensuremath{v}\xspace}

\renewcommand{\lt}{\ensuremath{\ell^t}\xspace}
\newcommand{\rt}{\ensuremath{r^t}\xspace}
\newcommand{\vit}{\ensuremath{v^t_i}\xspace}

\renewcommand{\c}[2]{\ensuremath{s|_{#1}^{#2}\xspace}}

\newcommand{\bidin}[3]{\ensuremath{{#1}\textrm{ bids in }[#2, #3]}}
\newcommand{\winin}[3]{\ensuremath{{#1}\textrm{ wins in }[#2, #3]}}
\newcommand{\cwinin}[3]{\ensuremath{\prob[{#1}\textrm{ wins in }[#2, #3]| \max_{j} b_j \leq #3]}}

\newcommand{\allbidless}[1]{\ensuremath{\max_{j}b_{j} < #1}}
\newcommand{\ind}{\ensuremath{\mathbb{I}\xspace}}
\newtheorem{thm}{Theorem}[section]
\newtheorem{lemma}[thm]{Lemma}
\newtheorem{corollary}[thm]{Corollary}

\newtheorem{observation}[thm]{Observation}

\newtheorem{fact}[thm]{Fact}

\addauthor{ab}{red}
\addauthor{ym}{cyan}
\addauthor{jm}{magenta}

\usepackage{enumitem}
\usepackage[compact]{titlesec}
\titlespacing{\section}{10pt}{10pt}{10pt}
\titlespacing{\subsection}{10pt}{10pt}{10pt}

\makeatletter
\newenvironment{proof*}[1][\proofname]{\par
  \pushQED{\qed}%
  \normalfont \partopsep=\z@skip \topsep=\z@skip
  \trivlist
  \item[\hskip\labelsep
        \itshape
    #1\@addpunct{.}]\ignorespaces
}{%
  \popQED\endtrivlist\@endpefalse
}
\makeatother

\newcommand{\poly}{\textrm{poly}}
\newcommand{\ibot}{\ensuremath{\gamma\xspace}}
\newcommand{\pbot}[1]{\ensuremath{p_{#1}\xspace}}
\newcommand{\terr}{\ensuremath{\epsilon}\xspace}
\newcommand{\lj}{{\ensuremath{\ell_{\tau+1}}}\xspace}
\renewcommand{\l}[1]{{\ensuremath{\ell_{#1}}}\xspace}
\newcommand{\ljo}{{\ensuremath{\ell_{\tau}}}\xspace}

\newcommand{\est}[1]{\ensuremath{\widehat{#1}}\xspace}
\newcommand{\add}{\ensuremath{\alpha}\xspace}
\newcommand{\mult}{\ensuremath{\mu}\xspace}
\newcommand{\inside}{\ensuremath{\texttt{Inside}}\xspace}
\newcommand{\intervals}{\ensuremath{\texttt{Intervals}}\xspace}
\newcommand{\iwin}{\ensuremath{\texttt{IWin}}\xspace}
\newcommand{\kaplan}{\ensuremath{\texttt{Kaplan}}\xspace}
\renewcommand{\bot}{0}
\newcommand{\yt}{y^t}
\newcommand{\km}{\ensuremath{\texttt{KM}}}
\title{Learning Valuation Distributions from Partial Observation}
\author[1]{Avrim Blum\thanks{Supported in part by the National Science Foundation under grants CCF-1101215, CCF-1116892, CCF-1331175, and IIS-1065251.  Email: {\tt avrim@cs.cmu.edu}}}

\author[2]{Yishay Mansour\thanks{This research was supported in part by The Israeli Centers of Research Excellence (I-CORE) program,
(Center  No. 4/11), by a grant from the Israel Science Foundation (ISF), by a grant from United
  States-Israel Binational Science Foundation (BSF), and by a grant
  from the Israeli Ministry of Science (MoS). Email: {\tt mansour@tau.ac.il}}}
\author[1]{Jamie Morgenstern\thanks{Supported in part by the National Science Foundation under grants  CCF-1116892, CCF-1331175 and IIS-1065251 and by a Simons Award for Graduate Students in Theoretical Computer Science.  Email: {\tt jamiemmt@cs.cmu.edu}}}

\affil[1]{Computer Science Department, Carnegie Mellon University}
\affil[2]{Tel Aviv University and Microsoft Research, Hertzelia}

\begin{document}
\maketitle

\begin{abstract}
  Auction theory traditionally assumes that bidders' valuation
  distributions are known to the auctioneer, such as in the
  celebrated, revenue-optimal Myerson
  auction~\citep{myerson}. However, this theory does not describe how
  the auctioneer comes to possess this information.  Recently,
  ~\citet{cole2014sample} showed that an approximation based on a
  finite sample of independent draws from each bidder's distribution
  is sufficient to produce a near-optimal auction.  In this work, we
  consider the problem of learning bidders' valuation distributions
  from much weaker forms of observations.  Specifically, we consider a
  setting where there is a repeated, sealed-bid auction with $n$
  bidders, but all we observe for each round is {\em who} won, but not
  how much they bid or paid.  We can also participate (i.e., submit a
  bid) ourselves, and observe when {\em we} win. From this
  information, our goal is to (approximately) recover the inherently
  recoverable part of the underlying bid distributions.  We also
  consider extensions where different {\em subsets} of bidders
  participate in each round, and where bidders' valuations have a
  common-value component added to their independent private values.

\end{abstract}

\thispagestyle{empty}
\newpage
\clearpage
\setcounter{page}{1}

\newpage
\section{Introduction}

Imagine that you get a call from your supervisor, who asks you to find
out how much various companies are bidding for banner advertisements
on a competitor's web site. She wants you to recover the distribution
of the bids for each one of the advertisers. Your boss might have many
reasons why she wants this information: to compare their bids there
and on your web site; to use as market research for opening a new web
site which would be attractive to some of those advertisers; or simply
to estimate the projected revenue.

This would be a trivial task if your competitor was willing to give
you this information, but this is unlikely to happen. Industrial
espionage is illegal, and definitely not within your expertise as a
computer scientist. So, you approach this task from the basics and
consider what you might observe. At best, you might be able to observe
the outcome for a particular auction, namely the winner, but
definitely not the price, and certainly not the bids of all
participants. There is, however, a way to observe more detailed
information: \emph{you can participate in a sequence of auctions and
  see whether or not you win!} If you lose with a bid $b$, you know
that the winner (and perhaps other bidders) bid more than $b$; if you
win with bid $b$, you know every other bidder bid less than $b$. In
general, we assume you will also observe the winner of the auction
explicitly (e.g., you can visit the webpage and view the banner ad of
the auction in question). Is this a strong enough set of tools to
recover the distributions over independent but not necessarily
identical bid distributions?

If only your boss had instead given you the task of estimating the
winning bid distribution in that auction, you would be able to
accomplish this easily. By inserting random bids, and observing their
probability of winning, you would be able to recover the distribution
of the winning bid. However, this was not the task you were assigned:
your boss wants the distribution of each bidder's bids, not just those
where they win the auction.

As a first attempt, your esteemed colleague suggests a trivial (and
completely incorrect) approach (which you do not even consider). As
before, you can submit random bids, and observe for each advertiser,
how many times he wins in auctions with your random bid. This will
estimate the distribution over bids he makes in auctions he
wins. However, when we condition on a bidder $i$ winning, we should
expect to see a sample which is skewed towards higher bids.  To see
that the distribution over winning bids is a poor estimation for the
distribution over bids for each bidder, consider the following
example. Suppose you can even observe the bid of the winner. There are
$n$ advertisers, each bidding uniformly in $[0,1]$.  The distribution
of the winning bid of a given advertiser would have an expectation of
$\frac{n}{n+1}$, whereas the expectation of his bid is $\frac{1}{2}$;
indeed, the distribution over winning bids would be a poor
approximation to his true bid distribution, namely, uniform in
$[0,1]$.  Additional complications arise with this approach when
advertisers are asymmetric, which is certainly the case in practice.

At this point, you decide to take a more formal approach, since the
simplest possible technique fails miserably. This leads you to the
following abstraction. There are $n$ bidders, where bidder $i$ has bid
distribution $\Di$. The $n$ bidders participate in a sequence of
auctions. In each auction, each bidder draws an independent bid
$b_i\sim \Di$ and submits it.\footnote{We remark that if the repeated
  auction is incentive-compatible the bid and valuation of the
  advertiser would be the same (and we use them interchangeably).  If
  this is not the case, then $\Di$ should be viewed as the
  distribution of bidder $i$'s {\em bids}.} We have the power to
submit a bid $b_0$, which is independent of the bids $b_i$, to the
auction. After each auction we observe the identity of the winner (but
nothing else about the bids).  Our goal is to construct a distribution
$\widehat{D}_i$ for each advertiser $i$ which is close to $\Di$ in
total variation.  Our main result in this work is to solve this
problem efficiently. Namely, we derive a polynomial time algorithm
(with polynomial sample complexity) that recovers an approximation
$\widehat{D}_i$ of each of the distributions $\Di$, down to some price
$p_\ibot$, below which there is at most $\ibot$ probability of any
bidder winning.\footnote{If the winning bid is never (or very rarely)
  below some price $p$, then we will not be able to learn
  approximations to the distributions $\Di$ below $p$.  For example,
  if bidder 1's distribution $\D_1$ has support only on $[\frac{1}{2},
  1]$ and bidder 2's distribution $\D_2$ has support only on $[0,
  \frac{1}{2})$, then since the winning bid is always at least
  $\frac{1}{2}$, we will never be able to learn anything about $\D_2$
  other than the fact that its support lies in
  $[0,\frac{1}{2})$. Thus, our goal will be to learn a good
  approximation to each $\Di$ only above a price $p_\ibot$ such that
  there is at least a $\ibot$ probability of the winning bid being
  below $p_\ibot$.}

Following your astonishing success in recovering the bid distributions
of the advertisers, your boss has a follow-up task for you. Not all
items for sale, or users to which these ads are being shown, are
created equal, and the advertisers receive various attributes
describing the user (item for sale) before they submit their
bid. Those attributes may include geographic location, language,
operating system, browser, as well as highly sensitive data that might
be collected though cookies. Your boss asks you to recover how the
advertisers bid as a function of those vectors.

For this more challenging task, we can still help, under the
assumption that we have access to these attributes for the observed
auctions, under some assumptions. We start with the assumption that
each bidder uses a linear function of the attributes for his
bid. Namely, let $x$ be the attribute vector of the user, then each
advertiser has a weight vector $w_i$ and his bid is $x\cdot w_i$. For
this case we are able to recover efficiently an approximation
$\widehat{w}_i$ of the weight vectors $w_i$.

A related task is to assume that the value (or bid) of an advertiser
has a common shared component plus a private value which is
stochastic.  Namely, given a user with attributes $x$, the shared
value is $x\cdot w$, where the $w$ is the same to all advertisers, and
each advertiser draws a private value $v_i\sim \Di$.  The bid of
advertiser $i$ is $x\cdot w +v_i$ The goal is to recover both the
shared weights $w$ as well as the individual distributions.  We do
this by ``reduction'' to the case of no attributes, by first
recovering an approximation $\widehat{w}$ for $w$, and then using it
to compute the common value for each user $x$.

One last extension we can handle focuses on who participates in the
auction. So far, we assumed that in each auction, all the advertisers
participate. However, this assumption is not really needed. Our
approach is flexible enough, such that if we received for each auction
the participants, this will be enough to recover the bidding
distributions for each bidder who shows up often enough. Note that if
there are $n$ advertisers and each time a random subset shows up, we
are unlikely to see the same subset show up twice; we can learn about
bidder $i$'s distribution over bids even when she is never competing
in the same context, assuming her bid distribution does not depend on
who else is bidding.

\subsection{Related Work}

Problems of reconstructing distributional information from limited or
censored observations have been studied in both the medical statistics
literature and the manufacturing/operations research literature. In
medical statistics, a basic setting where this problem arises is
estimating survival rates (the likelihood of death within $t$ years of
some medical procedure), when patients are continually dropping out of
the study, independently of their time of death.  The seminal work in
this area is the Kaplan-Meier product-limit
estimator~\citep{kaplanmeier}, analyzed in the limit in the original
paper and then for finite sample sizes in \cite{foldes1981}, see also
its use in \cite{darkpools}.  In the manufacturing literature, this
problem arises when a device, composed of multiple components, breaks
down when the first of its components breaks down. From the statistics
of when devices break down and which components failed, the goal is to
reconstruct the distributions of individual component lifetimes
\citep{N70,M81}.  The methods developed (and assumptions made, and
types of results shown) in each literature are different.  In our
work, we will build on the approach taken by the Kaplan-Meier
estimator (described in more detail in Section
\ref{sec:second-reserve}), as it is more flexible and better suited to
the types of guarantees we wish to achieve, extending it and using it
as a subroutine for the kinds of weak observations we work with.

The area of prior-free mechanism design has aimed to understand what
mechanisms achieve strong guarantees with limited (or no) information
about the priors of bidders, particularly in the area of revenue
maximization. There is a large variety truthful mechanisms that guarantee a constant
approximation (see, cf, \cite{HartlineKarlin07}). A different direction is adversarial
online setting which minimize the regret with respect to the best single price (see, \cite{KleinbergL03}),
or minimizing the regret for the reserve price of
a second price auction \cite{cesa-bianchi:regret}. In
\cite{cesa-bianchi:regret} it was assumed that bidders have an identical
bid distribution and the algorithm observes the actual sell price
after each auction, and based on this the bidding distribution is approximated. 

A recent line of work tries to bridge between the Bayesian setting and the adversarial one,
by assuming we observe a limited number of samples.
For a regular distribution, as single sample bidders' distributions
is sufficient to get a $1/2$-approximation to the optimal
revenue~\citep{dhangwatnotai2010revenue}, which follows from an
extension of the \citet{bulow} result that shows the revenue from a
second-price auction with $n+1$ (i.i.d) bidders is higher than the
revenue from running a revenue-optimal auction with $n$ bidders. 
%
%
%
Recent work of~\citet{cole2014sample} analyzes the
number of samples necessary to construct a $1-\epsilon$-approximately
revenue optimal mechanism for asymmetric bidders: they show it is
necessary and sufficient to take $\poly\left(\frac{1}{\epsilon},
  n\right)$ samples from each bidder's distribution to construct an
$1-\epsilon$-revenue-optimal auction for bid distributions are strongly regular.
We stress that in this work we do not make {\em any} assumptions about the bid distribution.

\citet{chawla2014data} design mechanisms which are approximately
revenue-optimal and also allow for good \emph{inference}: from a
sample of bids made in Bayes-Nash equilibrium, they would like to
reconstruct the distribution over values from which bidders are
drawn. This learning technique relies heavily on a sample being drawn
\emph{unconditionally} from the \emph{symmetric} bid distribution,
rather than only seeing the \emph{winner's identity} from
\emph{asymmetric} bid distributions, as we consider in this work.

We stress that in all the ``revenue maximization'' literature has a
fundamentally different objective than the one in this paper.  Namely,
our goal is to reconstruct the bidders' bid distributions, rather than
focusing of the revenue directly.  Our work differs from previous work
in this space in that it assumes very limited observational
information. Rather than assuming all $n$ bids as an observation from
a single run of the auction, or even observing only the price, 
we see only the identity of highest bidder. We do
not need to make any regularity assumption on the bid distribution
(monotone hazard rate, regular, etc.), our methodology handles {\em
  any} continuous bid distribution.\footnote{Note that we measure the
  distance between two distributions using the total variation
  distance, which is essentially ``additive''.}

\section{Model and Preliminaries}\label{section:model}
We assume there are $n$ bidders, and each $i\in [n]$ has some unknown
valuation distribution $\Di$ over the interval $[0,1]$. Each sample
$t\in [m]$ refers to a fresh draw $\vit\sim\Di$ for each $i$. The
label of sample $t$ will be denoted $\yt = \argmax_{i}\vit$, the
identity of the highest bidder. Our goal is to estimate \Fi, the
cumulative distribution for \Di, for each bidder $i$, up to $\epsilon$
additive error for all values in a given range. In Section
\ref{sec:extensions} we examine extensions and modifications to this
basic model.

We consider the problem of finding (sample and computationally)
efficient algorithms for constructing an estimate $\Fei$ of $\Fi$, the
cumulative distribution function, such that for all bidders $i$ and
price levels $p$, $\Fei(p) \in \{\Fi(p) \pm \epsilon\}$.  However, as
discussed above, this goal is too ambitious in two ways.  First, if
the labels contain no information about the value of bids, the best we
could hope to learn is the relative probability each person might win,
which is insufficient to uniquely identify the CDFs, even without
sampling error.  We address this issue by allowing, at each time $t$,
our learning algorithm to insert a fake bidder $0$ (or reserve) of
value $v^t_0 = \rt$; the label at time $t$ will be $\yt = \argmax_{i}
\vit$ ($\yt = 0$ will refer to a sample where the reserve was not met,
or the fake bidder won the auction).  The other issue, also described
above, is that there will be values below which we simply cannot
estimate the $\D_i$ since bids below that value do not win.  In
particular, if bids below price $p$ never win, then any two
cumulatives $\Fi, \Fi'$ that agree above $p$ will be statistically
indistinguishable. Thus, we will consider a slightly weaker goal. We
will guarantee our estimates $\Fei(p)\in \Fi(p)\pm \epsilon$ for all
$p$ where $\prob[\textrm{someone winning with a bid at most $p$}] \geq
\gamma$. Then, our goal is to minimize $m$, the number of samples
necessary, to do so, and we hope to have $m \in poly(n,
\frac{1}{\epsilon}, \frac{1}{\ibot})$, with high probability of
success over the draw of the sample.  One final (and necessary)
assumption we will make is that each $\Di$ has no point masses, and
our algorithm will be polynomial in the maximum slope $L$ of the
$\Fi$s.\footnote{This assumption is useful for this paper for two
  reasons. First, it allows us to eschew any issues associated with
  tie-breaking, since they happen with probability 0. Second, if there
  were no continuity assumption, there might be point masses. If we
  wished to have an additive accuracy guarantee as above, this would
  force our learning algorithm to be able to determine the
  \emph{exact} location of these point masses, which couldn't be done
  in polynomial time (for example, suppose $F_i$ had a point mass at
  $\sqrt{2}$).}.

\subsection{A brief primer on the Kaplan-Meier
  estimator}\label{sec:kaplan}

Our work is closely related in spirit to that of the Kaplan-Meier
estimator, \km, for survival time; in this section, we describe the
techniques used for constructing the \km~\citep{kaplanmeier}. This
will give some intuition for the estimator we present in
Section~\ref{sec:second-reserve}. We translate the results found
in~\citet{kaplanmeier} to an auction setting from the survival rate
literature. Suppose each sample $t$ is of the following form. Each
bidder $i$ draws their bid $b^t_i\sim \Di$ independently of each other
bid. The label $\yt = (\max_i b^t_i, \argmax_i b^t_i)$ consists of the
winning bid and the identity of the winner. From this, we would like
to reconstruct an estimate $\Fei$ of $\Fi$. Given $m$ samples, relabel
them so that the winning bids are in increasing order, e.g. $b^1_{i_1}
\leq b^2_{i_2} \leq b^m_{i_m}$. Here is some intuition behind the \km:
$\prob[b_i \leq x] = \prob[b_i \leq x| b_i \leq y] \cdot \prob[b_i
\leq y]$ for $y > x$. Repeatedly applying this, we can see that, for
$x < y_1 < y_2 <\cdots < y_r$,
\begingroup
\everymath{\scriptstyle}
\normalsize
\begin{align}
\begin{split}
\Fi(x) & =   \prob[b_i \leq x] = \prob[b_i \leq x | b_i \leq y_1]\; \prob[b_i
  \leq y_1 | b_i \leq y_2 ] \cdots \prob[b_i
  \leq y_{r-1} | b_i \leq y_r ]\; \prob[b_i \leq y_r]  \\
  & = \prob[b_i \leq x | b_i \leq y_1]\;\prob[b_i \leq y_r] \;\prod_{t =
    1}^r \prob[b_i \leq y_t | b_i \leq y_{t+1}]
\end{split}
\label{eq:bayes}
\end{align}
\endgroup
Now, we can employ the observation in Equation~\ref{eq:bayes}, if only
we knew how to convert the samples into estimates of such conditional
probabilities. Since other players' bids are independent, we can
estimate the conditional probabilities as follows:
\begin{align}
  \prob\left[ b_i \leq b^t_{i_t} | b_i \leq b^{t+1}_{i_{t+1}}\right] &
  \approx \begin{cases}
    \frac{t-1}{t} & \textrm{ if $i$ won sample $t$} \\
    1 & \textrm{ if $j \neq i$ won sample $t$}\\
\end{cases}\label{eq:conditional}
\end{align}
\noindent Thus, combining Equations~\ref{eq:bayes}
and~\ref{eq:conditional}, we have the Kaplan-Meier estimator:

\begingroup
\everymath{\scriptstyle}
\small
\[\km(x) = \prod_{t : b^t_j \geq x}\left(\frac{t - 1}{t}\right)^{\ind[i\textrm{ won sample } t]}\]
\endgroup

\noindent Our estimator is morally similar to \km, though it differs
in several important ways. First, and most importantly, we do not see
the winning bid explicitly; instead, we will just have lower or upper
bounds on the highest non-reserve bid (namely, the reserve bid when
someone wins or we win, respectively). Secondly, \km\ generally has no
control issue; in our setting, we are \emph{choosing} one of the
values which will censor our observation. We need to pick appropriate
reserves to get a good estimator (picking reserves that are too high
will censor too many observations, only giving us uninformative upper
bounds on bids, and reserves that are too low will never win, giving
us uninformative lower bounds on bids). Our estimator searches the
space $[0,1]$ for appropriate price points to use as reserves to
balance these concerns.

\section{Learning bidders' valuation
  distributions}\label{sec:second-reserve}
In this section, we assume we have the power to insert a reserve
price, and observe who won. Using this, we would like to reconstruct
the CDFs of each bidder $i$ up to some error, down to some price
$\pbot{i}$ where $i$ has probability no more than $\ibot$ of winning
at or below $\pbot{i}$, up to additive accuracy \terr. Our basic plan
of attack is as follows. We start by estimating the probability $i$
wins with a bid in some range $[a, a+\delta]$, by setting reserve
prices at $a$ and $a+\delta$, and measuring the difference in
empirical probability that $i$ wins with the two reserves. We then
estimate the probability that no bidder bids above $a+\delta$ (by
setting a reserve of $a+\delta$ and observing the empirical
probability that no one wins). These together will be enough to
estimate the probability that $i$ wins with a bid in that range,
conditioned on no one bidding above the range. We then show, for a
small enough range, this is a good estimate for the probability $i$
bids in the range, conditioned on no one bidding above the
range. Then, we chain these estimates together to form \kaplan, our
estimator.

More specifically, to make this work we select a partition of $[0,1]$
into a collection of intervals. This partition should have the
following property. Within each interval $[x, y]$, there should be
probability at most $\beta$ of any person bidding in $[x,y]$,
conditioned on no one bidding above $y$. This won't be possible for
the lowest interval, but will be true for the other intervals. Then,
the algorithm estimates the probability $i$ will win in $[x,y]$
conditioned on all bidders bidding at most $y$. This then $(1-\beta)$
(multiplicatively) approximates the probability $i$ bids in $[x,y]$
(conditioned on all bidders bidding less than $y$). Then, the
algorithm combines these estimates in a way such that the
approximation factors do not blow up to reconstruct the CDF.

\begin{algorithm}\label{alg:kaplan}
\KwData{$\epsilon, \ibot, \delta, L$, where $L$ is the Lipschitz constant of the $F_i$s}
\KwResult{$\est{F}_i$}
Let $\est{F}_i(0) = 0$,  $\est{F}_i(1) = 1$,  $k = \frac{2Ln}{\beta\ibot} + 1$,  $\delta' = \frac{\delta}{3k(\log{k} + 1)}$, $\beta = \frac{\epsilon\ibot}{32 nL}$, $\add = \beta^2/96$,  $\mult = \beta/96$, $T =  \frac{8\ln 6/\delta'}{\alpha^2\gamma^2\left(\frac{\mu}{2}\right)^2}$\;
Let $\l{1}, \ldots, \l{k'} = \intervals(\beta,
\ibot, T)$\;
\For{$t= 2$ to $k'-1$} {
  Let $r_{\l{\tau}, \l{\tau+1}} = \iwin(i, \l{\tau}, \l{\tau+1}, T)$\;
}
\For{$t= 2$ to $k' - 1$} {
  Let $\est{F}_i(\l{\tau}) = \prod_{\tau'\geq t+1}(1- r_{\l{\tau'}, \l{\tau'+1}})$\;
}
Define $\est{F}_i(x) = \max_{\l{\tau} \leq x}\est{F}_i(\l{\tau})$\;
\caption{\kaplan, estimates the CDF of $i$ from samples with reserves}
\end{algorithm}

\begin{thm}\label{thm:kaplan}
  With probability at least $1-\delta$, \kaplan outputs $\est{F}_i$, an estimate of $F_i$, with sample complexity
  \[m = O\left(\frac{n^8L^8\ln\frac{nL}{\epsilon\ibot}\left(\ln\frac{1}{\delta} + \ln\ln\frac{nL}{\epsilon\ibot}\right)}{\ibot^{10}\epsilon^6}\right)\]
\noindent   and, for all $p$ where $\prob[\exists~j \textrm{s.t. $j$ wins with a
    bid }\leq p] \geq \ibot$, if each CDF is $L$-Lipschitz, the error
  is at most:

\[F_i(p) - \epsilon \leq \est{F}_i(p) \leq F_i(p) + \epsilon.\]
\end{thm}

\kaplan calls several other functions, which we will now informally
describe, and state several Lemmas describing their guarantees (the
formal definitions can be found in Figure~\ref{fig:helpers} and the
proofs can be found in Appendix~\ref{sec:inequalities}).  \iwin
estimates the probability $i$ wins in the region $[\l{\tau},
\l{\tau+1}]$, conditioned on all bids being at most $\l{\tau+1}$.
\intervals partitions $[0,1]$ into small enough intervals such that,
conditioned on all bids being in or below that interval, the
probability of any bidder bidding within the interval is small.
(Essentially, $\l{2}$ is $p_\gamma$, and therefore we are not interested
in the estimation in $[0,\l{2}]$, and by definition $\l{1}=0$.)

Here are three lemmas which will be useful in the proof of
Theorem~\ref{thm:kaplan}. Lemma~\ref{lem:iwin} bounds the number of
samples \iwin uses and bounds the error of its
estimate. Lemma~\ref{lem:intervals} does similarly for \intervals.
Lemma~\ref{lem:winbid} states that, if a region $[\l{\tau},
\l{\tau+1}]$ is small enough, the probability that $i$ bids in
$[\l{\tau}, \l{\tau+1}]$ (conditioned on all bids being at most
$\l{\tau+1}$) is well-approximated by the probability that $i$ wins
with a bid in $[\l{\tau}, \l{\tau+1}]$ (conditioned on all bids being
at most $\l{\tau+1}$). In combination, these three imply a guarantee
on the sample complexity and accuracy of estimating
$\cwinin{i}{\l{\tau}}{\l{\tau+1}}$, which is the key ingredient of the
\kaplan estimator.

\begin{lemma}\label{lem:iwin}
  Suppose, for a fixed interval $[\ljo, \lj]$,
  $\prob[\winin{i}{0}{\lj}]\geq \ibot$. Then, \iwin($i$, $\ljo, \lj,
  T$) outputs $p^i_{\ljo, \lj}$ such that

\[ (1-\mult)\cwinin{i}{\ljo}{\lj} - \add \leq p^i_{\ljo, \lj} \leq (1+\mult)\cwinin{i}{\ljo}{\lj} + \add,\]

\noindent with probability at least $1-3\delta'$ and uses $3T$ samples, for the
values of $T, \delta'$ as in \kaplan.
\end{lemma}

\begin{lemma}\label{lem:intervals}
Let $T$ as in \kaplan. Then, $\intervals(\beta, \ibot, T,
  L, n)$ returns $0 = \l{1} < \cdots < \l{k} = 1$ such that
\begin{enumerate}
\item $k\leq \frac{48Ln}{\beta\ibot}$
\item For each $\tau\in [2, k]$,
$\prob[\max_{j}b_j\in [\l{\tau}, \l{\tau+1}]|\max_{j}b_j \leq \l{\tau+1}] \leq \frac{\beta}{16}$
\item $\prob[\max_{j}b_j \in [\l{1}, \l{2}]]\leq \ibot$
\end{enumerate}

\noindent with probability at least $1-3k\log(k)\delta'$, when bidders' CDFs
are $L$-Lipschitz, using at most $3kT\log{k}$ samples.
\end{lemma}

With the guarantee of Lemma~\ref{lem:intervals}, we know that the
partition of $[0,1]$ returned by \intervals is ``fine enough''. Now,
Lemma~\ref{lem:winbid} shows that, when the partition fine enough, the
conditional probability $i$ wins with a bid in each interval is a good
estimate for the conditional probability $i$ bids within that
interval.

\begin{lemma}\label{lem:winbid}
  Suppose that, for bidder $i$ and some $ 0 \leq \ljo \leq \lj \leq
  1$,

  \[\prob[\max_{j\neq i} b_j \in[\ljo, \lj] | \max_{j\neq i} b_j <
  \lj] \leq \beta.\]
\noindent Then,

\[ 1 \geq \frac{\prob[\winin{i}{\ljo}{\lj}|\allbidless{\lj}]}{\prob[\bidin{i}{\ljo}{\lj}|\allbidless{\lj}]} \geq 1-\beta \]
\end{lemma}

Finally, we observe that $F_i$ can be written as the product of
conditional probabilities

\begin{observation}\label{obs:kaplan}
  Consider some set of points $0 < \l{1} < \ldots < \l{k} = 1$.
  $F_i(\l{\tau})$ can be rewritten as the following product:
\[ F_i(\l{\tau-1}) = F_i(\l{\tau})(1-\prob[b_i \geq \l{\tau-1} | b_i \leq \l{\tau}]) = \prod_{\tau' \geq t} (1-\prob[b_i\geq\l{\tau'-1}| b_i \leq \l{\tau'}]) =\prod_{\tau' \geq t} (1-\prob[b_i\in[\l{\tau'-1}, \l{\tau'}]| b_i \leq \l{\tau'}]) \]
\end{observation}

With these pieces in place, we prove Theorem~\ref{thm:kaplan}.

\begin{proof*}[Proof of Theorem~\ref{thm:kaplan}]
  Notice that there are at most $k'$ events each of which happens with
  probability at most $\delta' = \frac{\delta}{k'}$ (namely, that
  \intervals returns a poor partition, or for each interval, of which
  there are at most $k'-1$, by Lemma~\ref{lem:intervals}, that \iwin
  is not accurate as described by Lemma~\ref{lem:iwin}). Thus, by a
  union bound, none of these events occur with probability
  $1-\delta$. Thus, for the remainder of the proof we assume the
  partition returned by \intervals is good and each call to \iwin is
  accurate.

  It will suffice to prove, for the lattice points in our
  discretization, that \kaplan provides an $\epsilon$-approximation to
  the CDF. This follows because
\begin{align*}
  F_i(\l{\tau}) - F_i(\l{\tau-1}) &= \prob[\bidin{i}{\l{\tau-1}}{\l{\tau}}]\\
  & = \prob[\bidin{i}{\l{\tau-1}}{\l{\tau}}|b_i \leq \l{\tau}]\\
  & \leq \prob[\bidin{i}{\l{\tau-1}}{\l{\tau}}|\max_j b_j \leq \l{\tau}]\\
  & \leq (1+\beta) \prob[\winin{i}{\l{\tau-1}}{\l{\tau}}|\max_j b_j \leq \l{\tau}]\\
  & \leq (1+\beta) \beta = \beta + \beta^2 \leq \frac{\epsilon}{2}
\end{align*}
\noindent where the third and fourth inequality follows from
Lemma~\ref{lem:intervals} and Lemma~\ref{lem:winbid}, and the final
one from the fact that $\beta < \frac{\epsilon}{4}$. Thus, our lattice is
fine enough that it suffices to show accuracy of the lattice
points. We start by rewriting $F_i(\l{\tau})$, using
Observation~\ref{obs:kaplan}:
\begin{align}
\begin{split}
  F_i(\l{\tau}) & = \prod_{\tau' \geq \tau + 1} (1-\prob[b_i\in[\l{\tau'-1} ,
  \l{\tau'}]\mid b_i\leq \l{\tau'}]).
\end{split}\label{eq:kp}
\end{align}
So, one can compute the probability of bidding at most $\l{\tau-1}$ by
multiplying together a collection of probabilities of bidding within
intervals above $\l{\tau}$.  Let the event $\max_j b_j \leq \l{\tau'}$
be denoted $M_{\l{\tau'}}$.  Now, we can apply
Lemma~\ref{lem:intervals} to imply that, for all $\tau'$,

\[\prob[\max_j b_j\in [\l{\tau'-1}, \l{\tau'}]| M_{\l{\tau'}}] \leq  \frac{\beta}{16} = \beta'\]

\noindent which, by Lemma~\ref{lem:winbid}, implies for all $\tau'$ that
\begin{align}
1 & \geq
 \frac{\prob[\winin{i}{\l{\tau'-1}}{\l{\tau'}}|M_{\l{\tau'}}]}{\prob[\bidin{i}{\l{\tau'-1}}{\l{\tau'}}|m_{\l{\tau'}}]} =  \frac{\prob[\winin{i}{\l{\tau'-1}}{\l{\tau'}}|M_{\l{\tau'}}]}{\prob[\bidin{i}{\l{\tau'-1}}{\l{\tau'}}|\bidin{i}{0}{\l{\tau'}}]} \geq 1-\beta' \label{eq:winbid}
\end{align}
\noindent where the equality comes from the independence of the bids. Then,
combining Equations~(\ref{eq:winbid}) and~(\ref{eq:kp}), we know
\begin{align*}
\begin{split}
&\prod_{\tau'\geq \tau + 1} \left(1 - \prob[\winin{i}{\l{\tau'-1}}{\l{\tau'}}|M_{\l{\tau'}}]\right) =  F_i(\l{\tau}) \geq \prod_{\tau'\geq \tau + 1} (1 - \frac{\prob[\winin{i}{\l{\tau'-1}}{\l{\tau'}}|M_{\l{\tau'}}]}{1-\beta'})
\end{split}\label{eq:bound}
\end{align*}
Then, by Fact~\ref{fact:divmult},

\[F_i(\l{\tau}) \in \left[
  \prod_{\tau' \geq \tau + 1}
  \left(1 -  (1+2\beta')\prob[\winin{i}{\l{\tau'-1}}{\l{\tau'}}|M_{\l{\tau'}}]\right),
\prod_{\tau' \geq \tau + 1}\left(1 - \prob[\winin{i}{\l{\tau'-1}}{\l{\tau'}}|M_{\l{\tau'}}]\right)\right]\]

\noindent Now, Lemma~\ref{lem:iwin} states that the result of \iwin are correct
within an additive \add and multiplicative \mult, thus
\begin{align*}
 &\prod_{\tau' \geq \tau + 1}(1 -
  (1+\mult)\prob[\winin{i}{\l{\tau'-1}}{\l{\tau'}}|M_{\l{\tau'}}] -\add) \leq \est{F}_i(\l{\tau}) \leq
\prod_{\tau' \geq \tau + 1}(1 -
  (1-\mult)\prob[\winin{i}{\l{\tau'-1}}{\l{\tau'}}|M_{\l{\tau'}}]+\add).
\end{align*}
Now, we simply need to look at the potential difference in these
terms. We will consider the lower bound on $F_i(\l{\tau})$ and upper
bound on $\est{F}_i(\l{\tau})$ (the other direction is analogous).
\begin{align*}
& \prod_{\tau' \geq \tau + 1}(1 -
  (1-\mult)\prob\left[\winin{i}{\l{\tau'-1}}{\l{\tau'}}|M_{\l{\tau'}}\right]+\add)  -
 \prod_{\tau' \geq t +1}(1 -
  (1+2\beta')\prob\left[\winin{i}{\l{\tau'-1}}{\l{\tau'}}|M_{\l{\tau'}}\right])\\
& \leq  \prod_{\tau' \geq \tau + 1}(1 -
  \prob\left[\winin{i}{\l{\tau'-1}}{\l{\tau'}}|M_{\l{\tau'}}\right]+\mult\beta'+\add)  -
 \prod_{\tau' \geq \tau + 1}(1 -
  \prob\left[\winin{i}{\l{\tau'-1}}{\l{\tau'}}|M_{\l{\tau'}}\right]-2\beta'^2)\\
& \leq  \prod_{\tau' \geq \tau + 1}(1 -
  \prob\left[\winin{i}{\l{\tau'-1}}{\l{\tau'}}|M_{\l{\tau'}}\right]+\beta'^2)  -
 \prod_{\tau' \geq \tau + 1}(1 -
  \prob\left[\winin{i}{\l{\tau'-1}}{\l{\tau'}}|M_{\l{\tau'}}\right]-2\beta'^2)\\
& \leq  \prod_{\tau' \geq \tau + 1}(1-2\beta'^2)(1 -
  \prob\left[\winin{i}{\l{\tau'-1}}{\l{\tau'}}|M_{\l{\tau'}}\right]) -
 \prod_{\tau' \geq \tau +1}(1+2\beta'^2)(1 -
  \prob\left[\winin{i}{\l{\tau'-1}}{\l{\tau'}}||M_{\l{\tau'}}\right])\\
& \leq (1-2\beta'^2)^k \prod_{\tau' \geq \tau + 1}(1 -
  \prob\left[\winin{i}{\l{\tau'-1}}{\l{\tau'}}|M_{\l{\tau'}}\right]) -
 (1+4\beta'^2)^k\prod_{\tau' \geq \tau + 1}(1 -
  \prob\left[\winin{i}{\l{\tau'-1}}{\l{\tau'}}|M_{\l{\tau'}}\right])\\
& \leq (1-4k\beta'^2) \prod_{\tau' \geq \tau + 1}(1 -
  \prob\left[\winin{i}{\l{\tau'-1}}{\l{\tau'}}|M_{\l{\tau'}}\right]) -
 (1+8k\beta'^2)\prod_{\tau' \geq t +1}(1 -
  \prob\left[\winin{i}{\l{\tau'-1}}{\l{\tau'}}|M_{\l{\tau'}}\right])\\
&\leq 12k\beta'^2 \leq 12\frac{16Ln}{\beta\ibot}\beta'^2 \leq \frac{3 Ln\beta}{\ibot}
\leq \frac{\epsilon}{2}
\end{align*}
\noindent where the first follows from
$\prob\left[\winin{i}{\l{\tau'-1}}{\l{\tau'}}|\allbidless{\l{\tau'}}\right]
\leq \beta'$, the second by the definition of $\add =
\frac{\beta'^2}{2}$, $\mult = \frac{\beta'}{2}$, the third again, by
$\prob\left[\winin{i}{\l{\tau'-1}}{\l{\tau'}}|M_{\l{\tau'}}\right]
\leq \beta'$, the fourth from $2\beta' < \frac{1}{2}$, the fifth and
sixth from basic algebra, the seventh by the bound on $k \leq \frac{16
  Ln}{\beta\ibot}$, by Lemma~\ref{lem:intervals}, the eighth by $\beta'
= \frac{\beta}{16}$, and the ninth by $\beta = \frac{\epsilon\ibot}{32
  nL}$.

The sample complexity bound and failure probability follow from
Lemmas~\ref{lem:intervals} and~\ref{lem:iwin}, substituting in for
various parameters, since \iwin is called $k$ times. Thus, in total,
there are $\leq 3k\log(k) + 3k$ empirical estimates made, each with
probability at most $\delta'$ of failure, each with sample size $T$.
\end{proof*}

\begin{figure}[!ht]
\begin{algorithm}[H]
\KwData{$ \ljo, \lj, T$}
\KwResult{$p^{\in}_{\ljo, \lj}$}
Let $S_1$ be a sample of size $T$ with reserve $\ljo$\;
Let $S_2$ be a sample of size $T$ with reserve $\lj$\;
Return $p^{\in}_{\ljo, \lj} =1 - \frac{\sum_{t\in S_2}\ind[\bot\textrm{ wins }t]}{\sum_{t\in S_1}\ind[\bot\textrm{ wins }t]}$\;
\caption{\inside, estimates $\prob[\max_j b_j \geq \ljo |\max_j b_j \leq \lj]$}\label{alg:inside}
\end{algorithm}

\begin{algorithm}[H]
  \KwData{$i$, $\ljo, \lj$, $T$}
  \KwResult{$p^i_{\ljo, \lj}$}
\caption{\iwin, Estimates $\prob[\winin{i}{\ljo}{\lj}|\max_{j} b_j < \lj]$}\label{alg:intervalest}
Let $S_\ljo$ be a sample with reserve $\lj$ of size $T$\;
Let $S_\lj$ be a sample with reserve $\ljo$ of size $T$\;
Let $S_{cond}$ be a sample with reserve $\lj$ of size $T$\;
Output $p^i_{\l{\tau}, \l{\tau+1}} = \frac{\sum_{t\in S_\ljo}\ind[i\textrm{ wins on sample }t] - \sum_{t\in S_\lj}\ind[i\textrm{ wins on sample }t]}{\sum_{t\in S_{cond}}\ind[\bot\textrm{ wins on sample }t]}$\;
\end{algorithm}

\begin{algorithm}[H]\label{alg:intervals}
\KwData{$\beta, \ibot, T, n, L$}
\KwResult{$0 = \l{1} < \ldots < \l{k} = 1$}
Let $\l{k} = 1$, $c= k$, $p^i_{\l{c}} = 1$\;
\While(\tcp*[h]{Do binary search for the bottom of the next interval})
{$p^i_{\l{c}}>\ibot/2$} {
  Let $\est{\l{b}} = 0$\;
  \While(\tcp*[h]{The interval is too large}){$\inside(\est{\l{b}}, \l{c}, T)   > \frac{\beta}{48}$} {
    $\est{\l{b}} = \frac{\l{c} + \est{\l{b}}}{2}$\;
  }
  $\l{c-1} = \est{\l{b}}$\;
  $c = c - 1$\;
  Let $S_1$ be a sample of size $T$ with reserve $\l{c-1}$\;
  $p_{\l{c}} = \frac{\sum_{t\in S_1} \ind[j\geq 1 \textrm{ wins on sample }t]}{T}$\;
}
Return $0, \l{c}, \ldots, \l{k}$\;
\caption{\intervals, finds a partition of the bid space into regions where we estimate $f_i$}
\end{algorithm}

\caption{Helper functions}\label{fig:helpers}
\end{figure}

\subsection{Subsets}
The argument above extends directly to a more general scenario in
which not all bidders necessarily show up each time, and instead there
is some distribution over $2^{[k]}$ over which bidders show up each
time the auction is run.  As mentioned above, this is quite natural in
settings where bidders are companies that may or may not need the
auctioned resource at any given time, or keyword auctions where there
is a distribution over keywords, and companies only participate in the
auction of keywords that are relevant to them.  To handle this case,
we simply apply Algorithm \ref{alg:kaplan} to just the subset of time
steps in which bidder $i$ showed up when learning $\est{F}_i$.  We use
the fact here that even though the distribution over subsets of
bidders that show up need not be a product distribution (e.g., certain
bidders may tend to show up together), the maximum bid value of the
other bidders who show up with bidder $i$ is a random variable that is
independent of bidder $i$'s bid.  Thus all the above arguments extend
directly.  The sample complexity bound of Theorem \ref{thm:kaplan} is
now a sample complexity on observations of bidder $i$ (and so requires
roughly a $1/q$ blowup in total sample complexity to learn the
distribution for a bidder that shows up only a $q$ fraction of the
time).

\newcommand{\distx}{\ensuremath{\mathcal{P}}\xspace}
\section{Extensions and Other Models}\label{sec:extensions}


So far we have been in the usual model of independent private values.
That is, on each run of the auction, bidder $i$'s value is $v_i \sim
\Di$, drawn independently from the other $v_j$.  We now consider
models motivated by settings where we have different items being
auctioned on each round, such as different cameras, cars, or laptops,
and these items have observable properties, or features, that affect
their value to each bidder.

In the first (easier) model we consider, each bidder $i$ has its own
private weight vector $w_i \in R^d$ (which we don't see), and each
item is a feature vector $x \in R^d$ (which we do see).  The value for
bidder $i$ on item $x$ is $w_i \cdot x$, and the winner is the highest
bidder $\argmax_i w_i \cdot x$.  There is a distribution $\distx$ over
items, but no additional private randomness.  Our goal, from
submitting bids and observing the identity of the winner, is to learn
estimates $\tilde{w}_i$ that approximate the true $w_i$ in the sense
that for random $x \sim \distx$, with probability $\geq 1-\epsilon$,
the $\tilde{w}_i$ correctly predict the winner and how much the winner
values the item $x$ up to $\pm \epsilon$.

In the second model we consider, there is just a single common vector
$w$, but we reintroduce the distributions $\Di$.  In particular, the
value of bidder $i$ on item $x$ is $w \cdot x + v_i$ where $v_i \sim
\Di$.  The ``$w \cdot x$'' portion can be viewed as a common value due
to the intrinsic worth of the object, and if $w = \vec{0}$ then this
reduces to the setting studied in previous sections.  The goal of the
algorithm is to learn both the common vector $w$ and all the $\Di$.

The common generalization of the above two models, with different
unknown vectors $w_i$ and unknown distributions $\Di$ appears to be
quite a bit more difficult (in part because the expected value of a
draw from $\Di$ conditioned on bidder $i$ winning depends on the
vector $x$).  We leave as an open problem to resolve learnability
(positively or negatively) in such a model.  We assume that $\|x\|_2
\leq 1$ and $\|w_i\|_2\leq 1$, and as before, all valuations are in
$[0,1]$.

\subsection{Private value vectors without private randomness}
Here we present an algorithm for the setting where each bidder $i$ has
its own private vector $w_i \in R^d$, and its value for an item $x \in
R^d$ is $w_i \cdot x$.  There is a distribution $\distx$ over items,
and our goal, from submitting bids and observing the identity of the
winner, is to accurately predict the winner and the winning bid.
Specifically, we prove the following:

\begin{thm}
With probability $\geq 1-\delta$, the algorithm below using sample size

\begingroup
\everymath{\scriptstyle}
\small
$$m = O\left(\frac{1}{\epsilon^2}\left[dn^2 \log(1/\epsilon) + \log(1/\delta)\right]\right)
$$
\endgroup
produces $\tilde{w}_i$ such that on a $1-\epsilon$ probability mass of
$x \sim \distx$ we have $i^* \equiv \argmax_i \tilde{w}_i \cdot x =
\argmax_i w_i \cdot x$ (i.e., a correct prediction of the winner), and
additionally have $|\tilde{w}_{i^*} \cdot x - w_{i^*} \cdot x| \leq
\epsilon$.
\end{thm}

\begin{proof*}
  Our algorithm is simple.  We will participate in $m$ auctions using
  bids chosen uniformly at random from $\{0, \epsilon, 2\epsilon,
  \ldots, 1\}$.  We observe the winners, then solve for a consistent
  set of $\tilde{w}_i$ using linear programming.  Specifically, for
  $t=1,\ldots, m$, if bidder $i_t$ wins item $x_t$ for which we bid
  $b_t$, then we have linear inequalities:
\begin{eqnarray*}
\tilde{w}_{i_t} \cdot x_t & > & \tilde{w}_j \cdot x_t \;\;\;\;(\forall j \neq i_t)\\
\tilde{w}_{i_t} \cdot x_t & > & b_t.
\end{eqnarray*}
Similarly, if we win the item, we have:
\begin{eqnarray*}
b_t & > & \tilde{w}_j \cdot x_t \;\;\;\; (\forall j).
\end{eqnarray*}
Let $\distx^*$ denote the distribution over pairs $(x,b)$ induced by
drawing $x$ from $\distx$ and $b$ uniformly at random from $\{0,
\epsilon, 2\epsilon, \ldots, 1\}$ and consider a $(k+1)$-valued target
function $f^*$ that given a pair $(x,b)$ outputs an integer in
$\{0,1,\ldots, n\}$ indicating the winner (with 0 indicating that our
bid $b$ wins).  By design, the vectors $\tilde{w}_1, \ldots,
\tilde{w}_n$ solved for above yield the correct answer (the correct
highest bidder) on all $m$ pairs $(x,b)$ in our training sample.  We
argue below that $m$ is sufficiently large so that by a standard
sample complexity analysis, with probability at least $1-\delta$, the
true error rate of the vectors $\tilde{w}_i$ under $\distx^*$ is at
most $\epsilon^2/(1+\epsilon)$.  This in particular implies that for
at least a $(1-\epsilon)$ probability mass of items $x$ under
$\distx$, the vectors $\tilde{w}_i$ predict the correct winner for
{\em all} $\frac{1+\epsilon}{\epsilon}$ bids $b \in \{0, \epsilon,
2\epsilon, \ldots, 1\}$ (by Markov's inequality).  This implies that
for this $(1-\epsilon)$ probability mass of items $x$, not only do the
$\tilde{w}_i$ correctly predict the winning bidder but they also
correctly predict the winning bid value up to $\pm \epsilon$ as
desired.

Finally, we argue the bound on $m$.  Any given set of $n$ vectors
$\tilde{w}_1, \ldots, \tilde{w}_n$ induces a $(n+1)$-way partition of
the $(d+1)$-dimensional space of pairs $(x,b)$ based on which of $\{0,
\ldots, n\}$ will be the winner (with 0 indicating that $b$ wins).
Each element of the partition is a convex region defined by
halfspaces, and in particular there are only $O(n^2)$ hyperplane
boundaries, one for each pair of regions.  Therefore, the total number
of ways of partitioning $m$ data-points is at most $O(m^{(d+1)n^2})$.
The result then follows by standard VC upper bounds for desired error
rate $\epsilon^2/(1+\epsilon)$.
\end{proof*}

\subsection{Common value vectors with private randomness}
\newcommand{\epp}{\epsilon'}

We now consider the case that there is just a single common vector
$w$, but we reintroduce the distributions $\Di$.  In particular, there
is some distribution $\distx$ over $x \in R^d$, and the value of
bidder $i$ on item $x$ is $w \cdot x + v_i$ where $v_i \sim \Di$.  As
before, we assume that $\|x\|_2 \leq 1$ and $\|w_i\|_2\leq 1$, and all
valuations are in $[0,1]$.  The goal of the algorithm is to learn both
the common vector $w$ and all the $\Di$.  We show here how we can
solve this problem by first learning a good approximation $\tilde{w}$
to $w$ which then allows us to reduce to the problem of Section
\ref{sec:second-reserve}.  In particular, given parameter $\epp$, we
will learn $\tilde{w}$ such that
$$\Pr_{x \sim \distx}\left( |w \cdot x - \tilde{w} \cdot x| \leq \epp \right) \geq 1-\epp.$$
\noindent Once we learn such a $\tilde{w}$, we can reduce to the case of Section
\ref{sec:second-reserve} as follows: every time the algorithm of
Section \ref{sec:second-reserve} queries with some reserve bid $b$, we
submit instead the bid $b + \tilde{w} \cdot x$.  The outcome of this
query now matches the setting of independent private values, but where
(due to the slight error in $\tilde{w}$) after the $v_i$ are each
drawn from $\D_i$, there is some small random fluctuation that is
added (and an $\epp$ fraction of the time, there is a large
fluctuation).  But since we can make $\epp$ as polynomially small as
we want, this becomes a vanishing term in the independent private
values analysis.  Thus, it suffices to learn a good approximation
$\tilde{w}$ to $w$, which we do as follows.

\begin{thm}
  With probability $\geq 1-\delta$, the algorithm below using running
  time and sample size polynomial in $d$, $n$, $1/\epp$, and
  $\log(1/\delta)$, produces $\tilde{w}$ such that $\Pr_{x \sim
    \distx}[|\tilde{w} \cdot x - w \cdot x| \leq \epp] \geq 1-\epp$.
\end{thm}

\begin{proof*}
  Let $\D_{max}$ denote the distribution over $\max[v_1,...,v_n]$.  By
  performing an additive offset, specifically, by adding a new feature
  $x_0$ that is always equal to 1 and setting the corresponding weight
  $w_0$ to be the mean value of $\D_{max}$, we may assume without loss
  of generality from now on that $\D_{max}$ has mean value
  $0$.\footnote{Adding such an $x_0$ and $w_0$ has the effect of
    modifying
each $v_i$ to $v_i-E[v_{max}]$.  The resulting distributions over $w \cdot x + v_i$ are
all the same as before, but now $\D_{max}$ has a zero mean value.}

Now, consider the following distribution over labeled examples
$(x,y)$.  We draw $x$ at random from $\distx$.  To produce the label
$y$, we bid a uniform random value in $[0,1]$ and set $y=1$ if we lose
and $y=0$ if we win (we ignore the identity of the winner when we
lose).  The key point here is that if the highest bidder for some item
$x$ bid a value $b \in [0,1]$, then with probability $b$ we lose and
set $y=1$ and with probability $1-b$ we win and set $y=0$.  So,
$\exp[y] = b$.  Moreover, since $b = w \cdot x + v_{max}$, where
$v_{max}$ is picked from $\D_{max}$ which has mean value of 0, we have
$\exp[b|x] = w \cdot x$.  So, $\exp[y|x] = w \cdot x$.

So, we have examples $x$ with labels in $\{0,1\}$ such that $\exp[y|x]
= w \cdot x$.  This implies that $w \cdot x$ is the predictor of
minimum squared loss over this distribution on labeled examples (in
fact, it minimizes mean squared error for every point $x$).  Moreover,
any real-valued predictor $h(x) = \tilde{w} \cdot x$ that satisfies
the condition that $\exp_{(x,y)}[(\tilde{w} \cdot x - y)^2] \leq
\exp_{(x,y)}[(w \cdot x - y)^2] + \epp^3$ must satisfy the condition:
$$\Pr_{x \sim \distx}\left( |w \cdot x - \tilde{w} \cdot x| \leq \epp \right) \geq 1-\epp.$$
This is because a predictor that fails this condition incurs an
additional squared loss of $\epp^2$ on at least an $\epp$ probability
mass of the points.  Finally, since all losses are bounded (we know
all values $w \cdot x$ are bounded since we have assumed all
valuations are in $[0,1]$, so we can restrict to $\tilde{w}$ such that
$\tilde{w}\cdot x$ are all bounded), standard confidence bounds imply
that minimizing mean squared error over a sufficiently (polynomially)
large sample will achieve the desired near-optimal squared loss over
the underlying distribution.
\end{proof*}

\bibliographystyle{plainnat}
\bibliography{sources}

\appendix
\section{Inequalities}\label{sec:inequalities}

\begin{lemma}\label{lem:conditional}
  Suppose $X$ is observable and $Y$ is observable, and assume that
  $\prob[Y]\geq \ibot$. 
 Using $2T$ samples, with probability $1-\delta$, we can estimate $\prob[X|Y] =
  \frac{\prob[X\cap Y]}{\prob[Y]}$ buy $\hat{p}$ such that
 \[
 \prob[X|Y]-\alpha-\mu\leq
 (1-\mult)\prob[X| Y] - \add \leq \hat{p} \leq
  (1+\mult)\prob[X| Y] + \add\leq \prob[X|Y]+\alpha+\mu,\]
\end{lemma}

As a direct corollary, we know that \inside is a close approximation
to the quantity it estimates.
\begin{corollary}\label{cor:inside}
  $\inside(\l{\tau}, \l{\tau+1}, T)$ outputs an estimator $p^{\in}_{\ljo,
    \lj}$, such that, for $T$ as in \kaplan,

 \[(1-\mult)\prob[\max_j b_j
  \geq \ljo| \max_j b_j \leq \lj] - \add\leq p^{\in}_{\ljo, \lj} \leq
  (1+\mult)\prob[\max_j b_j \geq \ljo| \max_j b_j \leq \lj] + \add\]
and uses $2T$ samples.
\end{corollary}

Now, we prove Lemma~\ref{lem:winbid}, which is also a corollary of
Lemma~\ref{lem:conditional}.
\begin{proof}[Proof of Lemma~\ref{lem:winbid}]
Let, for a fixed $i, \l{\tau}, \l{\tau+1}$, the event that
  $\bidin{i}{\l{\tau}}{\l{\tau+1}}$ be denoted by $X$, the event that
  $\winin{i}{\l{\tau}}{\l{\tau+1}}$ be denoted by $Y$, and the event that
  $\allbidless{\l{\tau+1}}$ be denoted by $C$.

With this notation, we have an estimate of $\prob[Y|C]$ and want an
estimate of $\prob[X|C]$.
\begin{align*}
\prob[Y | C] & = \prob[ X | C] \times \prob[ Y | C, X]\\
& \geq  \prob[ X | C] \times \prob[ \textrm{everyone but $i$ bids $< \l{\tau}$}| C, X]\\
& =  \prob[ X | C] \times \prob[ \textrm{everyone but $i$ bids $< \l{\tau}$}| C]\\
& \geq  \prob[ X | C] \times (1- \beta)
\end{align*}

The first equality comes from the fact that $Y \subseteq X$, the next
inequality comes from the fact that, conditioned on $C$ and $X$, everyone
but $i$ bids $< \l{\tau}$ is a subset of $Y$ (the times when $i$ will win),
the next equality comes from the fact that $i$'s bid and $j$'s bid are
independent, and the final inequality follows from the assumption
 $\prob[\max_{j\neq i} b_j < \l{\tau} | \max_{j\neq i} b_j <\l{\tau+1}]
\geq 1 - \beta$.
\end{proof}

\begin{fact}\label{fact:divmult}
  Suppose $x\geq 0$ and $0 < \eta < \frac{1}{2}$. Then $\frac{x}{1 + \eta} \geq
  (1-\eta)x$ and $\frac{x}{1-\eta}\leq (1+2\eta)x$.
\end{fact}

\begin{proof}[Proof of Fact~\ref{fact:divmult}]
We prove $\frac{x}{1 + \eta} \geq (1-\eta)x$ first.
\begin{align*}
\frac{x}{1+\eta}  = \frac{(1-\eta)x}{1 - \eta^2} \geq (1-\eta)x \;\;\;\;\;\;\;\;\;\;\;\;\;\; \textrm{(Since $1-\eta^2 < 1$)}
\end{align*}

Now, we prove $\frac{x}{1-\eta} \leq (1+2\eta)x$, for $\eta\leq 1/2$. We have,
\begin{align*}
\frac{x}{1-\eta} & = x\sum_{i=0}^\infty \eta^i = x\left(1+\eta(\sum_{i=0}^\infty \eta^i)\right)\leq (1+2\eta)x,
\end{align*}
where the inequality follows from the fact that for $\eta\leq 1/2$ we have $\sum_{i=0}^\infty \eta^i =\frac{1}{1-\eta}\leq 2$.
\end{proof}
%
%
%
%
%
%
%


\begin{proof}[Proof of Lemma~\ref{lem:iwin}]
  We start by showing that, with no sampling error, the calculation
  $p^i_{x,y}$ we do is equivalent to $q^i_{x, y} = \prob[b_i \in [x,
  y] \wedge b_i > \max_{j\neq i} b_j| \max_j b_j < y]$.  When $x = y$,
  we will denote this simply as $q^i_x$ (similarly, $p^i_x$).
  Similarly, let $q^\bot_x$ denote the probability that no one wins
  when the reserve bidder is set to bid $x$ (and $p^\bot_x$ the
  empirical probability therein).

 By definition,

\begin{align*}
q^i_{\l{\tau}, \l{\tau+1}} &= \prob[b_i \in [\l{\tau}, \l{\tau+1}] \wedge b_i > \max_{j\neq i} b_j| \max_j b_j < \l{\tau+1}]\\
& = \frac{\prob[b_i \in [\l{\tau}, \l{\tau+1}] \wedge b_i > \max_{j\neq i} b_j \wedge \max_j b_j < \l{\tau+1}]}{\prob[\max_j b_j < \l{\tau+1}]}\\
& = \frac{\prob[b_i \in [\l{\tau}, \l{\tau+1}] \wedge b_i > \max_{j\neq i} b_j]}{\prob[\max_j b_j < \l{\tau+1}]} & \textrm{($i$ winning in $[\l{\tau}, \l{\tau+1}]$ implies $\max_j b_j < \l{\tau+1}$)} \\
& = \frac{\prob[b_i \geq \l{\tau} \wedge b_i > \max_{j\neq i} b_j] - \prob[b_i \geq \l{\tau+1} \wedge b_i > \max_{j\neq i} b_j]}{\prob[\max_j b_j < \l{\tau+1}]} \\
& = \frac{\prob[i \textrm{ wins with reserve } \l{\tau}] - \prob[i \textrm{ wins with reserve } \l{\tau+1}]}{\prob[\max_j b_j < \l{\tau+1}]} & \textrm{(Assuming no point masses, there are no ties)}\\
& = \frac{q^i_{\l{\tau},1} - q^i_{\l{\tau+1},1}}{q^\bot_{\l{\tau+1}}}
\end{align*}

The final form is identical to the estimated quantity used by \iwin.
It now suffices to now show that each of the three samples give us
good estimates of their respective true probabilities. A basic
Chernoff bound implies

  \[\prob[|p^i_{x,1} - q^i_{x, 1}| \geq \frac{\add\ibot(1-\mult)}{4}]\leq 2e^{-T\frac{1}{8}t_1\add^2\ibot^2(1-\mult)^2}.\]
Substituting $T = \frac{8\ln 6/\delta'}{\alpha^2\gamma^2\left(\frac{\mu}{2}\right)^2}$, and
  noting $\mult < 1- \mult$, we have
\[\prob[|p^i_{x,1} - q^i_{x, 1}| \geq \frac{\add\ibot(1-\mult)}{4}]\leq \delta'\]
for each of $x=\l{\tau}, \l{\tau+1}$. Similarly,

\[\prob[|p^\bot_x - q^\bot_x|> \frac{\mult\ibot}{2}]\leq 2e^{-\frac{T}{2} \mult^2\ibot^2}\]
and substituting for $T$, we have that $|p^{\bot}_{\l{\tau+1}} -
q^{\bot}_{\l{\tau+1}}| \geq \frac{\mult\ibot}{2}$ with probability at
most $\delta'$.  Thus, using a union bound, we have that with probability
at least $1-3\delta'$, for a particular $t$,

\begin{align}
\frac{q^{i}_{\l{\tau},1} - q^{i}_{\l{\tau+1},1} - \frac{\add\ibot(1-\mult)}{2}}{q^{\bot}_{\l{\tau+1}}+\frac{\mult\ibot}{2}}\leq \frac{p^{i}_{\l{\tau},1}- p^{i}_{\l{\tau+1},1}}{p^\bot_{\l{\tau+1}}}\leq \frac{q^{i}_{\l{\tau},1} - q^{i}_{\l{\tau+1},1} + \frac{\add\ibot(1-\mult)}{2}}{q^{\bot}_{\l{\tau+1}}-\frac{\mult\ibot}{2}}\label{eqn:total}
\end{align}

Now, it suffices to show that Equation~(\ref{eqn:total}) implies the
relative error stated previously. By assumption, $p^{i}_{0,\l{\tau+1}} >
\ibot$. This implies that the probability everyone bids at most $\l{\tau+1}$
is at least $\ibot$ (for a winning bid of $\l{\tau+1}$ to win, all bids must
be at most $\l{\tau+1}$), so

\begin{align}
q^{\bot}_{\l{\tau+1}} \geq \ibot.\label{eqn:pbot}
\end{align}

Then,

\begin{align*}
   p^{i}_{\l{\tau},\l{\tau+1}}
& =
  \frac{p^{i}_{\l{\tau},1} - p^{i}_{\l{\tau+1},{1}}}{p^{\bot}_{\l{\tau+1}}}\\
& \geq
  \frac{q^{i}_{\l{\tau},1} - q^{i}_{\l{\tau+1},{1}} -
    \frac{1}{2}\add\ibot(1-\mult)}{q^{\bot}_{\l{\tau+1}}+\frac{\mult\ibot}{2}}\\
 & \geq
  \frac{q^{i}_{\l{\tau},1} - q^{i}_{\l{\tau+1}, 1} -
    \add\ibot}{q^{\bot}_{\l{\tau+1}}+\mult\ibot} & \textrm{(Since $\frac{(1-\mult)}{2} < 1$)}\\
  &\geq  \frac{q^{i}_{\l{\tau}, 1} - q^{i}_{\l{\tau+1}, 1} -
    \add\ibot}{q^{\bot}_{\l{\tau+1}}+\mult q^{\bot}_{\l{\tau+1}}} & \textrm{(By Eq.~(\ref{eqn:pbot}))}\\
 &=  \frac{q^{i}_{\l{\tau},1} - q^{i}_{\l{\tau+1},1} - \add\ibot}
 {q^{\bot}_{\l{\tau+1}}(1+\mult)}\\
  & = \frac{q^{i}_{\l{\tau},\l{\tau+1}}}{1+\mult} - \frac{ \add\ibot}{q^{\bot}_{\l{\tau+1}}(1+\mult)}\\
  & \geq \frac{q^{i}_{\l{\tau},\l{\tau+1}}}{1+\mult} - \frac{ \add}{(1+\mult)} & \textrm{(By Eq.~(\ref{eqn:pbot}))}\\
  & \geq \frac{q^{i}_{\l{\tau},\l{\tau+1}}}{1+\mult} - \add\\
  & \geq (1-\mult)q^{i}_{\l{\tau},\l{\tau+1}} - \add & \textrm{(By Fact~\ref{fact:divmult})}
\end{align*}

Now, we prove the upper bound on our estimator.

\begin{align*}
 p^{i}_{\l{\tau},\l{\tau+1}}
& =
  \frac{p^{i}_{\l{\tau},1} - p^{i}_{\l{\tau+1},{1}}}{p^{\bot}_{\l{\tau+1}}}\\
 &\leq \frac{q^{i}_{\l{\tau}, 1} - q^{i}_{\l{\tau+1}, 1} + \frac{(1-\mult)}{2}\add\ibot}{q^{\bot}_{\l{\tau+1}}-\frac{\mult\ibot}{2}}\\
& \leq \frac{q^{i}_{\l{\tau}, 1} - q^{i}_{\l{\tau+1}, 1} + (1-\mult)\add\ibot}{q^{\bot}_{\l{\tau+1}}-\frac{\mult\ibot}{2}}\\
& \leq\frac{q^{i}_{\l{\tau}, 1} - q^{i}_{\l{\tau+1}, 1} + (1-\mult)\add\ibot}{q^{\bot}_{\l{\tau+1}}-\frac{\mult q^{\bot}_{\l{\tau+1}, \l{\tau+1}}}{2}} & \textrm{(By Eq.~(\ref{eqn:pbot}))}\\
 & \leq \frac{q^{i}_{\l{\tau}, \l{\tau+1}}}{1-\frac{\mult}{2}} + \frac{(1-\mult)\add\ibot}{q^{\bot}_{\l{\tau+1}}(1-\frac{\mult}{2})}\\
& \leq \frac{q^{i}_{\l{\tau}, \l{\tau+1}}}{1-\frac{\mult}{2}} + \frac{\add\ibot}{q^{\bot}_{\l{\tau+1}}}\\
& \leq \frac{q^{i}_{\l{\tau}, \l{\tau+1}}}{1-\frac{\mult}{2}} + \add & \textrm{(By Eq.~(\ref{eqn:pbot}))}\\
& \leq (1+2\frac{\mult}{2})q^{i}_{\l{\tau}, \l{\tau+1}} + \add & \textrm{(By Fact.~\ref{fact:divmult})}\\
& = (1+\mult)q^{i}_{\l{\tau}, \l{\tau+1}} + \add
\end{align*}

Thus, both the upper and lower bounds on the estimator hold with
probability $1-\delta$.
\end{proof}

\begin{proof}[Proof of Lemma~\ref{lem:intervals}]
We will show each of the three parts to be true.
\begin{enumerate}
\item We start by proving that \intervals will output a partition with
  at most $\frac{24nL}{\beta\ibot}$ intervals.  We claim that each
  interval is at least $\frac{\beta\ibot}{24nL}$ in length, implying
  the above bound on the total number of intervals.

  Consider some current upper bound for an interval $\l{\tau+1}$.  If
  \intervals accepts some point $\l{\tau}$ such that $\l{\tau+1} - \l{\tau}
  \geq \frac{\beta\ibot}{24nL}$, then the bound trivially holds.

  If this does not hold, \intervals tests some point $\est{\l{\tau}}$ such that

\[\frac{\beta\ibot}{24nL} \geq \l{\tau+1} - \est{\l{\tau}} \geq \frac{\beta\ibot}{48nL}\]

since it is doing binary search. We claim \intervals will accept
$\est{\l{\tau}}$; if this is the case, the interval will have length at
least $\frac{\beta\ibot}{48nL}$. Notice that

\[ \prob[\max_j b_j\in  [\est{\l{\tau}}, \l{\tau+1}] | \max_j b_j \leq \l{\tau+1}]]
\leq  \prob[\max_j b_j\in  [\l{\tau+1} - \frac{\beta\ibot}{24nL}, \l{\tau+1}] | \max_j b_j \leq \l{\tau+1}]]\]

so it will suffice to show that \intervals would accept the smallest
possible value of $\est{\l{\tau}}$ (since that region will have the most
probability mass).  We bound the ratio, for a given $\l{\tau+1}$ such
that

\[\prob[\max_j b_j \in [\l{\tau+1} - \frac{\beta\ibot}{24nL}, \l{\tau+1}] | \max_j b_j \leq \l{\tau+1}] = \frac{\prob[\max_j b_j \leq \l{\tau+1} -  \frac{\beta\ibot}{24nL} ]}{\prob[\max_j b_j \leq \l{\tau+1}]}\]

for some upper point of an interval $\l{\tau+1}$ such that
$\prob[i\textrm{ wins with a bid }\leq \l{\tau+1}] \geq \ibot$.  Since
$F_j$ is $L$-Lipschitz for all $j$,

\[\prob[b_j \leq \l{\tau+1}] - \prob[b_j \leq \l{\tau+1}-\frac{\beta\ibot}{24nL} ] \leq L \frac{\beta\ibot}{24Ln}= \frac{\beta\ibot}{24n}.\]

Then, by summing this probability over all $n$ bidders, we have

\[\prob[\max_j b_j \leq \l{\tau+1}] - \prob[\max_j b_j \leq \l{\tau+1}-\frac{\beta\ibot}{24nL}] \leq \frac{\beta\ibot}{24} .\]

Rearranging terms, we have

\[\frac{\prob[\max_j b_j \leq \l{\tau+1}- \frac{\beta\ibot}{24nL}]}{ \prob[\max_j b_j \leq \l{\tau+1}]} \geq 1- \frac{\beta'\ibot}{\prob[\max_j b_j \leq \l{\tau+1}]} \geq 1 - \frac{\beta}{24}
\]

where the last inequality came from the fact that $\prob[i\textrm{
  wins with a bid }\leq \l{\tau+1}] \geq \prob[\max_j b_j \leq \l{\tau+1}]
\geq \ibot$. So, \intervals will accept $\est{\l{\tau}}$ as $\l{\tau}$, so
long as the empirical estimate of \inside is correct up to $\add +
\mult = \frac{\beta}{48}$, which is the case by
Corollary~\ref{cor:inside} with probability $1-3\delta'$.

\item

We now need to show

\[\prob[\max_j b_j \geq \l{\tau-1}|\max_j b_j \leq \l{\tau}] \leq \frac{\beta}{16}\]

holds for the lattice points $t > 3$.  Since $\prob[\max_j b_j \leq
\l{3} ] \geq \ibot$, by Corollary~\ref{cor:inside}, the accuracy
guarantee holds with probability $1 - 3\delta'$ for a fixed $t$ (since
$\add = \frac{\beta^2}{96}, \mult = \frac{\beta}{96}$, and the
condition by which $\l{\tau-1}$ was accepted was that the empirical
estimate of the above quantity was at most $\frac{\beta}{24}$). Thus,
with probability $1-3k\delta'$, the above holds for all $t>3$.

\item
  We begin by showing $\prob[\max_j b_j \leq \l{2} ] \leq \gamma$ with
  probability at least $1-\delta'$. The condition for stopping the
  search for new interval points is

  \[J = \frac{\sum_{t\in S_1}\ind[i\textrm{ wins on sample }t]}{T} \leq
  \frac{\gamma}{2}\]

  where $S_1$ is a random sample of size $T$ with reserve $\l{1}$. A
  basic Chernoff bound shows that

\[\prob[|J - \prob[\max_j b_j \leq \l{1} ]|\geq \frac{\ibot}{2}] \leq 2e^{-\frac{T\ibot^2}{2}}\]

which, for $T = \frac{32\ln\frac{6}{\delta'}}{\add^2\ibot^2\mult^2}$
is at most $\delta'$, so $\prob_{S_1}[\prob[\max_j b_j \leq \l{2} ]
\leq \gamma] \geq 1- \delta'$, as desired.

\end{enumerate}

It remains to sum up the total error probability and sample
complexity. The lower bound on the length of each interval also
implies a bound on the total number of empirical estimates made to
find a fixed $\l{\tau}$. Formally, the halving algorithm beginning with a
search space of size $\l{\tau+1}\leq 1$ will halt before the remaining
search space has shrunk to $\frac{\beta\ibot}{48Ln}$, which will take
at most $\log\frac{48Ln}{\beta\ibot} = \log(k)$ attempted interval
endpoints per accepted interval endpoint. Each of these attempts calls
\inside, which takes 2 estimates. For each accepted interval, an
estimate of the remaining probability mass is done. Thus, in total,
there are $2k\log(k) + k$ estimates done by \intervals. Each fails
with probability at most $\delta'$, so \intervals succeeds with
probability at least $1-3k\log(k)\delta'$ and uses at most
$3k\log(k)T$ samples.
\end{proof}

\end{document}